\documentclass[12pt,showkeys,superscriptaddress,nofootinbib]{revtex4}
\usepackage{amsmath,amssymb,euscript,yfonts,psfrag,latexsym,dsfont,graphicx,bbm,color,amstext,wasysym,mathrsfs,soul,framed}
\usepackage{comment}

\graphicspath{{/},{figures/}}

\newcommand{\mR}{{\mathbb R}}

\newcommand{\E}{{\mathbb E}}

\newcommand{\cP}{{\mathcal P}}

\newcommand{\R}{\mathbb{R}}

\newcommand{\D}{{\mathbb D}}

\newcommand{\visc}{{}^A\otJac}
\newcommand{\tsp}{\mathbf{T}}
\newcommand{\De}{\mathrm{d}}
\newcommand{\covdev}{\nabla^{W_2} }
\newcommand{\proj}{\mathrm{P}}
\newcommand{\otJac}{\mathrm{D}^{W_2}}
\newcommand{\otgrad}{\nabla^{W_2}}
\newcommand{\scrU}{\mathscr{U}}
\definecolor{llgrey}{rgb}{0.9,0.9,0.9}
\definecolor{lgrey}{rgb}{0.6,0.6,0.6}
\definecolor{lred}{rgb}{0.9,0.7,0.7}

\newtheorem{theorem}{Theorem}

\newtheorem{cor}{Corollary}
\newtheorem{lemma}{Lemma}

\begin{document}

\title{Extremal flows on Wasserstein space
}
\author{Giovanni Conforti}\affiliation{Centre de Math\'ematiques Appliqu\'ees, \'Ecole Polytechnique, Palaiseau, France}\email{giovanni.conforti@polytechnique.edu}
\homepage{https://sites.google.com/site/giovanniconfort/}

\author{Michele Pavon}\affiliation{Dipartimento di Matematica ``Tullio Levi-Civita", Universit\`a di Padova, via
Trieste 63, 35121 Padova, Italy}
\email{pavon@math.unipd.it}\homepage{http://www.math.unipd.it/~pavon}

\begin{abstract}
{\bf Abstract.} We develop an intrinsic geometric approach to calculus of variations on Wasserstein space. We show that the flows associated to the Schr\"{o}dinger bridge with general prior, to Optimal Mass Transport and to the Madelung fluid can all be characterized as annihilating the first variation of a suitable action. We then discuss the implications of this unified framework for stochastic mechanics: It entails, in particular, a sort of fluid-dynamic reconciliation between Bohm's and Nelson's stochastic mechanics. 
\end{abstract}

\maketitle
\section{Introduction}

Interpolation between two given probability distributions has important applications, for instance in image morphing. The  interpolationg flow of one-time marginals  $\{\mu_t; 0\le t\le 1\}$ may be thought of as a curve in  {\em Wasserstein space} \cite{Vil}\footnote{Vershik has argued \cite{Ver} that it should be called {\em Kantorovich space}. According to him, it was Dobrushin, unaware of Kantorovich's work, that spread this terminology.}. One example is the displacement interpolation flow associated to the Benamou-Brenier formulation of the {\em Optimal Mass Transport} (OMT) problem with quadratic cost \cite{BB}. Other examples are the flows associated to the {\em Schr\"{o}dinger Bridge Problem} (SBP) \cite{S1,S2} and the quantum evolution of a nonrelativistic particle in {\em Nelson's Stochastic Mechanics} (NSM) \cite{N1,N2}. It is known that SBP may be viewed as a ``regularization" of OMT , the latter problem being recovered through a ``zero-noise limit" \cite{Mik, mt, MT,leo,leo2,CGP3,CGP4}. Recently, it was shown by von Renesse \cite{vonR}, for the quantum mechanical {\em Madelung fluid} \cite{G}, and by one of us \cite{GC} for the Schr\"{o}dinger bridge, that their flows satisfy a suitable Newton-like law in Wasserstein space. In \cite{GNT}, the foundations of a Hamilton-Jacobi theory in Wasserstein space were laid. 

In this paper, we show that the solution flows  of OMT, SBP and NSM may all be seen as extremal curves in Wasserstein space of a suitable action. The actions only differ by the presence or the sign of a (relative) {\em Fisher information functional} besides the kinetic energy term. The solution marginals flows correspond to critical points, i.e. annihilate the first variation, of the respective functionals. The extremality condition implies indeed  a generalization of the known Newton-type second laws in Wasserstein space. With respect to \cite{GC}, we present an alternative approach.
There, the Newton's law for {SBP is proved} using its dual representation as a generalized $h$-transform of the equilibrium dynamics. Here, we use a more geometrical approach which consists in taking advantage of the Riemannian metric associated with the Wasserstein distance to implement the fundamental lemma of calculus of variations on the fluid dynamic formulation of SBP. This approach {has the nice feature of being more intrinsic, providing a unified framework for the results of \cite{vonR} and \cite{GC} and allows to prove some new ones. For example, we  are able to treat non-reversible prior laws for SBP and show that the non reversibility translates in the addition of a viscous term in Newton's law. The calculations of this article, however, are more or less formal, and more work is required to turn them into rigorous results. Let us note that for SBP rigorous results can be found in \cite{GC}.  Our results also shed light on the relation between various versions of stochastic mechanics (Schr\"{o}dinger, Bohm, Nelson), a controversial topic made more misterious by the presence of two different Lagrangians within Nelson's stochastic mechanics \cite{N2}.  

Some of the results presented here had been announced, without developing the geometric approach, in our conference article  \cite{CP1}. The paper is outlined as follows. In Section \ref{OMT}, we recall some some basic concepts and result from OMT. Section \ref{Schroedinger} is devoted to Schr\"{o}dinger bridges and entropic interpolation.  In Section \ref{SOC}, we introduce the basic geometric tools in Wasserstein space and provide an intrinsic variational representation for the entropic interpolating flow (Theorem \ref{FDR}). In Section \ref{VA}, we show that the entropic interpolation satisfies a Newton-type law involving the covariant derivative of the velocity generalizing to the case of a general non-reversible prior the results of \cite{GC}. Section \ref{MFSM} is devoted to the classical Madelung fluid and to stochastic mechanics. We first connect to the results of \cite{vonR} on the Madelung fluid. We then consider Nelson's stochastic mechanics and compare it to Bohmian mechanics. We finally briefly discuss the differences in stochastic mechanics between the fluid-dynamic variational principles and the stochastic variational principles.

\section{Background on optimal mass transport}\label{OMT}

The literature on OMT is huge and keeps growing at a fast pace. A selection of monographs and survey papers is \cite{RR,E,Vil,AGS,Vil2, AG,OPV, San}.  Let $\Pi(\mu_0,\mu_1)$ be ``couplings" of $\mu_0$ and $\mu_1$, namely probability distributions on  $\R^d\times \R^d$  with marginals $\mu_0$ and $\mu_1$. In the Kantorovich relaxed formulation of the original Monge problem, one seeks a distribution $\pi\in\Pi(\mu_0,\mu_1)$ which minimizes the index 
\[\int_{\R^N\times\R^N}c(x,y)d\pi(x,y),
\]
where $c$ represents the cost of transporting an infinitesimal ``pebble" from $x$ to $y$. When $c(x,y)=\|x-y\|^2$ and $\mu_0$, $\mu_1$ belong to $\mathcal P_2(\R^d)$, the set of probability measures $\mu$ on $\R^d$ with finite second moment, the Kantorovich-Rubinstein (Wasserstein) quadratic distance, is defined by
\begin{equation}\label{Wasserdist}
W_2(\mu_0,\mu_1)=\left(\inf_{\pi\in\Pi(\mu_0,\mu_1)}\int_{\R^d\times\R^d}\|x-y\|^2d\pi(x,y)\right)^{1/2}.
\end{equation}
As is well known \cite[Theorem 7.3]{Vil}, $W_2$ is a {\em bona fide} distance. Moreover, it provides a most natural way to  ``metrize" weak convergence\footnote{$\mu_k$ converges weakly to $\mu$ if $\int_{\R^N}fd\mu_k\rightarrow\int_{\R^N}fd\mu$ for every continuous, bounded function $f$.} in $\mathcal P_2(\R^N)$ \cite[Theorem 7.12]{Vil}, \cite[Proposition 7.1.5]{AGS}. The quadratic {\em Wasserstein space} $\mathcal W_2$ is defined as the metric space $\left(\mathcal P_2(\R^N),W_2\right)$. It is a {\em Polish space}, namely a separable, complete metric space.

A {\em dynamic} version of the above OMT problem was already {\em in fieri} in Gaspar Monge's 1781 {\em ``M\'emoire sur la th\'eorie des d\'eblais et des remblais"}. It was elegantly accomplished by Benamou and Brenier in \cite{BB} by showing that 
\begin{subequations}\label{eq:BB}
\begin{eqnarray}\label{BB1}&&W_2^2(\mu_0,\mu_1)=\inf_{(\rho,v)}\int_{0}^{1}\int_{\R^d}|v_t(x)|^2\rho_t(dx)dt,\\&&\partial_t \rho_t+\nabla\cdot(v_t\rho_t)=0,\label{BB2}\\&& \rho_0=\mu_0, \quad \rho_1=\mu_1.\label{boundary}
\end{eqnarray}\end{subequations}
Here the flow $\{\rho_t; 0\le t\le 1\}$ varies over continuous maps from $[0,1]$ to $\mathcal P_2(\R^d)$ and $v$ over vector fields.  In \cite[Chapter 7]{Vil2}, Villani provides motivation to study  the time-dependent version of OMT, namely that a time-dependent model gives a more complete description of the
transport and that the richer mathematical structure is useful. Further reasons are the following. It allows to view the optimal transport problem as an (atypical) optimal control problem \cite{CGP}-\cite{CGP4}.
It provides a ground on which the connection with the Schr\"{o}dinger bridge problem appears as a regularization of the former \cite{Mik, mt, MT,leo,leo2,CGP3,CGP4}. Similarly with Nelson's stochastic mechanics, see below. In some applications, such as interpolation of images \cite{CGP5} or spectral morphing \cite{JLG}, the interpolating flow is the object of interest.

Let $\{\mu_t; 0\le t\le 1\}$ be optimal for (\ref{eq:BB}). Then $(\mu_t)$ provides the {\em displacement interpolation} between $\mu_0$ and $\mu_1$ and may be viewed as a constant-speed geodesic joining $\mu_0$ and $\mu_1$ in Wasserstein space (Otto). 

\section{Schr\"{o}dinger bridges and entropic interpolation}\label{Schroedinger}

 Let $\Omega=C([0,1];\R^d)$ be the space of $\R^d$ valued continuous functions, $P$ a positive measure\footnote{It is sometimes convenient to work with infinite measures $P$. This is the case, for instance, when $P$ is  {\em stationary Wiener measure} (reversible Brownian motion) on $\R^d$} on $\Omega$. We recall the definition of the relative entropy functional 
\[ \D(Q\| P) = \begin{cases} \E_{Q}\left( \log \frac{\De Q}{\De P} \right), \quad & \mbox{if $ Q \ll P$}\\+ \infty \quad & \mbox{otherwise} \end{cases} \]
For given $P,\mu_0,\mu_1$, the \emph{Schr\"odinger Problem} is the problem of minimizing the relative entropy with respect to $P$ in the set ${\cal P}(\mu_0,\mu_1)$ of all laws whose marginals at times $t=0,1$ are $\mu_0$ and $\mu_1$, namely
\begin{equation}\label{problem}{\rm Minimize}\quad \D(Q\|P) \quad {\rm over} \quad Q\in{\cal P}(\mu_0,\mu_1).
\end{equation} 
The optimal solution is called the \emph{Schr\"odinger bridge} between $\mu_0$ and $\mu_1$ over $P$, and  its marginal flow $(\mu_t)$ is the \emph{entropic interpolation}.
There are at least three distinct situations when this problem arises naturally:
\begin{itemize}
\item[(a)] As a model for the ``hot gas experiment": This was Schr\"odinger's original motivation.
\item[(b)] As a statistical inference problem, in the entropic approach to Bayesian statistics developed by Jaynes, see e.g. \cite{Jaynes57,Jaynes82}.
\item[(c)] As a regularization of the Monge Kantorovich problem.
\end{itemize}
\paragraph{hot gas experiment} Here, the prior distribution $P$ is a stationary Langevin dynamics for a generator $\cal{L}$ of the form
\[ \mathcal{L} = \frac{\sigma}{2}\Delta - \nabla U \cdot \nabla, \]
which represents an \emph{equilibrium dynamics} for a particle system.  
 At time $t=0$ we are given $N$ independent particles  $(X^i_t)_{t \leq 1, i \leq N}$ whose configuration is $\mu$.
\[ \mu := \frac{1}{N} \sum_{i=1}^N \delta_{X^i_0}. \] 
We then let the particles travel independently following the Langevin dynamics and look at their configuration $\nu$ at $t=1$.
 \[ \nu:= \frac{1}{N}\sum_{i=1}^N \delta_{X^i_1} . \]
If the Langevin dynamics has good ergodic properties and $N$ is very large, one expects $\nu$ to be very close to the invariant measure of $P$, which we denote $\mathbf{m}$. However, although very unlikely, it is still possible to observe an \emph{unexpected configuration}, meaning that $\nu$ is significantly\footnote{Schr\"odinger writes in \cite{S2} \emph{`` un \'ecart spontan\'e et considerable par rapport \`a cette uniformit\'e" }}  different from $\mathbf{m}$. Schr\"odinger's question is to find, conditionally on this rare event, the most likely evolution of the particle system. In the limit when $N \rightarrow + \infty$, Sanov's theorem tells that the \emph{most likely} evolution is given by the minimiser of \eqref{problem} (see also \cite[sec.6]{leo}). This derivation of the Scr\"odinger problem can be viewed as a stochastic counterpart to the ``lazy gas experiment" of optimal transport \cite[Ch. 16]{Vil2}. Indeed
\begin{itemize}
	\item Particles choose their final destination minimizing the relative entropy instead of the mean square distance;
	\item Particles travel along Brownian bridges instead of geodesics.
\end{itemize}
We refer to \cite[sec.6]{leo} for an extensive treatment of this analogy.
Quite remarkably, the description of the hot gas experiment we gave is very similar to the original formulation of the problem that Schr\"odinger proposed back in 1932, (see \cite{S2}) when none of the modern tools of probability or of modern optimal transport theory was available.

\paragraph{Maximum entropy principle and statistical inference}
Suppose now that $P$ is a prior distribution on some hypothesis space. When measurements are made and new information becomes available, Jaynes' maximum entropy principle asserts that the prior distribution should be updated by picking, among all laws compatible with the information, the one which minimizes the relative entropy w.r.t. to $P$\footnote{The confusion between minimization and maximization comes from the different definitions of entropy. 
}. Loosely speaking, it is a ``maximum ignorance principle", saying that the posterior distribution should make the least possible prediction about everything which is beyond the available information. When the hypothesis space is $\Omega$, and the new information comes in the form of marginal laws, the problem of finding the posterior becomes SBP.

\paragraph{Regularization of Optimal transport}
Suppose now that the prior measure $P$ is {\em Markovian}. In this case, the classical results of Jamison \cite{Jam} imply that the solution of (\ref{problem})  is also Markovian. Thus, we can restrict our search to  ${\cal P}^M(\mu_0,\mu_1)$, namely Markovian measures in ${\cal P}(\mu_0,\mu_1)$.  Along the lines of   \cite[p.683]{CGP3}, \cite[Corollary 5.8]{GLR}, \cite{Leg} we get that the entropic flow of the Schr\"{o}dinger bridge can be characterized  through the following fluid dynamic problem
\begin{subequations}\label{FDproblem}
\begin{eqnarray}\label{SBB1}\inf_{(\rho,v)}\int_{t_0}^{t_1}\int_{\mR^d}\left[\frac{1}{2\sigma}|v_t(x)-v^P_t(x)|^2\right.\\\left.\nonumber+\frac{\sigma}{8}|\nabla\log\frac{\rho}{\rho^P}_t(x)|^2\right]\rho_t(x)dx dt,\\ \partial_t \rho_t+\nabla\cdot(v_t\rho_t)=0,\label{SBB2}\\ \rho_0=\nu_0, \quad \rho_1=\nu_1.\label{SBB3}
\end{eqnarray}
\end{subequations}
Here $v$ and $v^P$ are the {\em current velocity fields} \cite{N1} of the trial and prior measure, respectively.
Comparing (\ref{FDproblem}) to the OMT with prior formulated and studied in \cite{CGP3}, see also \cite{CGP4} for the Gauss-Markov case, we see that the essential difference is that there is here an extra term in the action functional
which has the form of a {\em relative Fisher information}.  To find the connection to the classical OMT, let us specialize to the situation where the prior $P=W^{\varepsilon}$. In that case, $v^P=u^P=0$ and, multiplying the criterion by $\varepsilon$, we get the problem
\begin{subequations}\label{FDproblem'}
\begin{eqnarray}\label{SBB1'}\inf_{(\rho,v)}\int_{t_0}^{t_1}\int_{\mR^d}\left[\frac{1}{2}|v_t(x)|^2+\frac{\varepsilon^2}{8}|\nabla\log\rho_t(x)|^2\right]\rho_t(x)dx dt,\\ \partial_t \rho_t+\nabla\cdot(v_t\rho_t)=0,\label{SBB2'}\\ \mu_0=\nu_0, \quad \mu_1=\nu_1.\label{SBB3'}
\end{eqnarray}
\end{subequations}
If $\varepsilon\searrow 0$ it appears that in the limit we get the Benamou-Brenier formulation of OMT (\ref{eq:BB}). This is indeed the case, see  \cite{Mik, mt, MT,leo,leo2} and  \cite{CGP3,CGP4} for the case with prior.

Alternatively, since both OMT and SBP admit a static formulation \cite{leo},
the \emph{regularization}  of OMT as in (\ref{Wasserdist}) can be obtained by adding a term proportional to the  Shannon entropy on $\mathbb{R}^d \times \mathbb{R}^d$
\begin{equation}\label{eq:regularizedOT}
 \inf_{\pi \in \Pi(\mu_0,\mu_1)}  \int_{\mathbb{R}^d \times \mathbb{R}^d }  |x-y|^2 \De \pi(x,y) + 2 \varepsilon \int_{\R^d \times \R^d} \pi(x,y) \log \pi(x,y) \De x\De y,
\end{equation}
where we used  the same notation for a measure and its density function. Problem \eqref{eq:regularizedOT} is the projection onto the marginals at $t=0,1$ of the Schr\"odinger problem for a prior $P$ which is stationary Wiener measure, i.e. the stationary law for the Markov generator
\[ \mathcal{L}^{\varepsilon} = \frac{\varepsilon}{2} \Delta. \]
It is well known that the solution to the path-space formulation of the Schr\"odinger problem \eqref{problem} can be recovered  by mixing the optimal coupling in \eqref{eq:regularizedOT} with Brownian bridges \cite{F2}, \cite[Sec.1]{leo}. For the computational relevance of (\ref{eq:regularizedOT}) see (\cite{Cuturi}). For the dynamic form of entropic regularization of OMT (\ref{FDproblem'}), discovered in \cite{CGP3}\footnote{Incidentally, this characterization of the Schr\"{o}dinger bridge flow answers at once a question posed by Carlen \cite[pp. 130-131]{Carlen}.}, see instead \cite{LYO}. Notice that it differs from the regularization of the Benamou-Brenier formulation which involves replacing the continuity equation with a Fokker-Planck equation \cite{CGP5}.

Conditions for existence and uniqueness for SBP and properties of the minimizing measure have been studied by many authors, most noticeably by Fortet, Beurlin, Jamison and F\"{o}llmer \cite{For,Beu,Jam,F2}.  If there is at least one $Q$ in ${\cal P}(\mu_0,\mu_1)$ such that
$\D(Q\|P)<\infty$, there exists a unique minimizer $Q^*$. Existence is guaranteed under conditions on $P$, $\rho_0$ and $\rho_1$, see \cite[Proposition 2.5]{leo2}. We shall tacitly assume henceforth that they are satisfied so that $Q^*$ is well defined. It has been observed since the early nineties that SBP can be turned, thanks to {\em Girsanov's theorem},  into a stochastic control problem with atypical boundary constraints, see \cite{DP,Bl,DPP,PW,FHS}.

\section{Second order calculus in Wasserstein space}\label{SOC}

It is known that the \emph{displacement interpolation} between two measures is a constant speed geodesic in the quadratic Wasserstein space $\mathcal W_2$. In a seminal paper \cite{O}, Otto showed that we can go beyond this. More precisely, it is possible to define formally a kind of Riemannian metric on $\cP_2(\R^d)$ for which the displacement interpolation is a constant speed geodesic.
We sketch here the main steps in the construction of this metric. All what follows is only formally correct, and much work is required to turn it into rigorous statements. For this, we refer the reader to \cite{AG}, \cite{Gigli}.
The first step consists in identifying the tangent space at $\mu$ with the space of square integrable vector fields.
\[ \tsp_{\mu} := \overline{\left\{ \nabla \varphi; \varphi \in C^{\infty}_c  \right\}}^{L^2(\mu)} \]
The second step is to define the first derivative (velocity field) $v_t \in \tsp_{\mu_t}$  of a curve $(\mu_t)$ through the continuity equation 
\[\partial_t \mu_t + \nabla \cdot ( v_t \mu_t)=0, \quad v_t \in \tsp_{\mu_t}. \]
 
Then, one defines the Riemannian metric by means of the $L^2$ product 
\[\langle \nabla \varphi, \nabla \psi \rangle_{\mathbf{T}_{\mu}} := \int_{\R^d} \langle \nabla \varphi(x), \nabla \psi (x)\rangle \, \mu(x) \De x, \]
where $\langle ., . \rangle$ stands for the standard inner product of $\R^d$. The Benamou-Brenier formula establishes that the displacement interpolation is a constant speed geodesic for this Riemannian structure, as it minimizes the energy functional among all curves with a given start and end. The next task is to define the covariant derivative $\covdev_{v_t} u_t$ of the vector field $(u_t)$ along the curve $(\mu_t)$ in such a way that the torsion free identity and the compatibility with the metric are satisfied. It turns out that to do so, the right definition is to set (cf. \cite[eq 6.7, Def 6.8]{AG})
\begin{equation}\label{eq:covdedv} \covdev_{v_t} u_t := \mathrm{P}_{\mu_t} \left( \partial_t u_t +  {\rm D}_{v_t} u_t \right) ,
 \end{equation} 
provided $u_t,v_t$ are smooth enough. Here, ${\rm D}_{v_t} u_t(x)$ denotes the Jacobian of $u_t$ at $x$ applied to $v_t$ and $\mathrm{P}_{\mu}: L^2(\mu) \rightarrow \tsp_{\mu}$ is the orthogonal projection. In particular, we have that the acceleration of a curve is given by
\begin{equation}\label{eq:acc} \covdev_{v_t} v_t = \partial_t v_t + \frac{1}{2} \nabla |v_t|^2.  
\end{equation}
Once the covariant derivative has been constructed, we can also define the Jacobian of a vector field. A vector field $\scrU$ on $\cP_2(\R^d)$ is a map 
 $ \mu \mapsto \scrU(\mu) \in\tsp_{\mu}$ that associates to $\mu\in \cP_2(\R^d)$ a gradient vector field over $\mathbb{R}^d$. We define the Jacobian by
\[ \otJac_v \scrU = \covdev_{v_t} \scrU(\mu_t)\Big|_{t=0}, \]
where $\mu_t$ is any curve such that $\mu_0= \mu$ and $v_0=v$
Finally, we deonte by $^A\otJac$ the antisymmetric part of the Jacobian, i.e.
\begin{equation}\label{eq:visc} \langle \visc_v \scrU,w \rangle_{\tsp_{\mu}} :=  \frac{1}{2}\left( \langle \otJac_v \scrU,w \rangle_{\tsp_{\mu}}-\langle \otJac_w \scrU,v \rangle_{\tsp_{\mu}} \right) 
\end{equation}
\paragraph*{Some classical functionals and their gradients.}
Consider a functional $\mathcal{F}:\cP_2(\R^d) \rightarrow \R$. Then, its gradient $\otgrad \mathcal{F} $ is defined via
\begin{equation}\label{eq:otgrad} \langle  \otgrad \mathcal{F},v \rangle_{\tsp_{\mu}} = \lim_{h \rightarrow 0} \frac{1}{h}\left( \mathcal{F}(\mu_{h}) - \mathcal{F}(\mu_0) \right), 
\end{equation}
where $(\mu_t)$ is any curve such that $\mu_0=\mu$ and $v_0=v$.

Let us recall some common functionals and compute their gradients. As we will see, they express the acceleration of the curves we study in this article.
\begin{itemize}
\item For a potential $U$, the corresponding energy functional is $\mathcal{E}_U(\mu)$ is
 \begin{equation}
 \label{eq:energy} \mathcal{E}_U(\mu)=\int_{\R^d} U(x) \mu(x) \De x .
  \end{equation}
Its gradient is easily computed. We have
\begin{equation}\label{eq:energygrad} \otgrad \mathcal{E}_U(\mu) = \nabla U  \end{equation}
However, one has to keep in mind that the right hand side still depends on $\mu$, because the properites of $\nabla U$ as a tangent vector on $\langle .,.\rangle_{\tsp_{\mu}}$ vary with $\mu$, although, as a vector field on $\R^d$, $\nabla U$ doesn't change.
 \item The Shannon Entropy functional is
 \begin{equation}\label{eq:shannon} \mathcal{S}(\mu) = - \int_{\R^d} \log \mu(x) \mu(x)dx \end{equation}
The gradient of $\mathcal{S}$ is known:
\begin{equation}\label{eq:shannongrad} \otgrad \mathcal{S}(\mu) = - \nabla \log \mu \end{equation}
\item The Fisher information is the squared norm of the gradient of the Shannon entropy, i.e.
\begin{equation}\label{eq:fish} \mathcal{I}(\mu) := \int_{\R^d} |\nabla \log  \mu |^2  \mu(x) dx \end{equation}	Its gradient is 
\begin{equation}\label{eq:fishgrad}    \otgrad \mathcal{I}(\mu) =-\nabla \left( |\nabla \log \mu|^2+2 |\Delta \log \mu| \right)  ,  
\end{equation}
see for instance \cite[pg. 12]{vonR}, \cite[Eq. (5)]{GC}. It should be observed that this relation can already be found in \cite[p.172]{Bo}.
\end{itemize}

\subsection{Time-symmetric fluid-dynamic formulation of SBP}\label{time-symmetric}

We proceed to derive below a characterization of the entropic interpolation  (\ref{FDproblem}) in the Riemannian-like geometry of Wasserstein space when the prior is any stationary Markov diffusion measure. Assume that the prior distribution $P$ is the stationary dynamics for a Markov generator of the form
\begin{equation}\label{browniandiff}	\mathcal{L} = \frac{\sigma}{2}\Delta + b \cdot \nabla, 
\end{equation}
where the vector field $b$  may not be of gradient type.
We obtain the following fluid-dynamical formulation
\begin{theorem}\label{FDR}
The entropic interpolation between $\mu_0,\mu_1$ is an optimal solution for the problem
\begin{subequations}\label{eq:FDproblemgio}
\begin{eqnarray}\label{FDproblem1gio}&&\inf_{(\rho,v)} \frac{1}{2}\int_{0}^1 |v_t - \frac{\sigma}{2}\otgrad \mathcal{S}(\rho_t) - \mathscr{B}(\rho_t)   |^2_{\mathbf{T}_{\rho_t}}  \De t	 \\&&\partial_t \rho_t+\nabla\cdot(v_t\rho_t)=0, \quad v_t \in \tsp_{\rho_t} \label{FDproblem2gio}\\&& \rho_0=\mu_0, \quad \rho_1=\mu_1,\label{FDproblem3gio}
\end{eqnarray}
\end{subequations}
where 
\begin{equation}\label{eq:calB} \mathscr{B}(\rho) :=\proj_{\rho}(b),
\end{equation} 
and  $\mathrm{P}_{\rho} : L^2(\rho) \rightarrow \mathbf{T}_{\rho}$ is the orthogonal projection and $\mathcal{S}$ the Shannon entropy \eqref{eq:shannon}.
\end{theorem}

\begin{proof}
It is known that the Schr\"odinger bridge is the law of a diffusion process $Q$ whose generator is of the same form as $P$, i.e.
\[ \frac{\sigma}{2}\Delta + c_t \cdot \nabla,  \]
for some time-dependent vector field $c_t(x)$. For such processes, we can rewrite the Kullback-Liebler divergence using It\^o calculus as
\begin{equation*}\label{entropyrep} \D( Q\|P) = \D( q_0 \| p_0 )+ \frac{1}{2}\E_{Q} \big( \int_{0}^{1} |(c_t-b)(\omega_t)|^2 \De t \big)= 
\frac{1}{2}\int_{0}^1\int_{\R^d} |c_t-b|^2(x) \rho_t(x) \De x \De t,
 \end{equation*}
 where $\rho_t$ is the marginal at time $t$ of $Q$. The marginal flow $(\rho_t)$ is known to satisfy the Fokker Planck equation
\begin{equation*}
\partial_t \rho_t = - \nabla \cdot ( c_t \rho_t ) + \frac{\sigma}{2} \Delta \rho_t,
\end{equation*}
Therefore we arrive at a first reformulation of the problem as
\begin{subequations}\label{eq:FDproblemgioaux1}
\begin{eqnarray}\label{FDproblem1gioaux1}&&	\inf_{(\rho,c)} \frac{1}{2}\int_{0}^1\int_{\R^d} |c_t-b|^2(x) \rho_t(x) \De x \De t, \\&&\partial_t \rho_t+\nabla\cdot\left((c_t -\frac{\sigma}{2}\nabla \log\rho_t )\rho_t\right)=0,\label{FDproblem2gioaux1}\\&& \rho_0=\mu_0, \quad \rho_1=\mu_1,\label{FDproblem3gioaux1}
\end{eqnarray}
\end{subequations}
Next, we claim that the latter is equivalent to
\begin{subequations}\label{eq:FDproblemgioaux2}
\begin{eqnarray}\label{FDproblem1gioaux2}&&	\inf_{(\rho,w)}\frac{1}{2}\int_{0}^1\int_{\R^d} |w_t -\proj_{\rho_t}(b)|^2(x) \rho_t(x) \De x \De t, \\&&\partial_t \rho_t+\nabla\cdot \left((w_t -\frac{\sigma}{2}\nabla \log\rho_t )\rho_t \right)=0, \quad w_t \in \tsp_{\rho_t}\label{FDproblem2gioaux2}\\&& \rho_0=\mu_0, \quad \rho_1=\mu_1,\label{FDproblem3gioaux2}
\end{eqnarray}
\end{subequations}
To see this, assume that $(\rho^*_t,w^*_t)$ is optimal for \eqref{eq:FDproblemgioaux2}. Then, it is not hard to show that if we set $c^*_t:= w^*_t +b-\proj_{\rho^*_t}(b)$, then $(\rho^*_t,c^*_t )$ is admissible and optimal for \eqref{eq:FDproblemgioaux1}. On the other hand it is known, see e.g. \cite{DP} that the optimal solution $(\rho^*_t,c^*_t)$ in \eqref{eq:FDproblemgioaux1} is such that $c^*_t-b$ is of gradient type, i.e. it belongs to $\tsp_{\rho^*_t}$. Thus, if we set $w^*_t := c^*_t -b + \proj_{\rho^*_t}(b)$                                                                                                                                                                                   we have that $(\rho^*_t,w^*_t)$ is admissible for \eqref{eq:FDproblemgioaux2}; after some standard calculations it is not hard to see that it is also optimal. Upon setting $v_t= w_t - \frac{\sigma}{2}\nabla \log \rho_t $, we obtain from \eqref{eq:FDproblemgioaux1} another equivalent formulation of \eqref{eq:FDproblemgioaux2} as
\begin{subequations}\label{eq:FDproblemgioaux3}
\begin{eqnarray}\label{FDproblem1gioaux3}&& \inf_{(\rho,v)}	\frac{1}{2}\int_{0}^1\int_{\R^d} |v_t + \frac{\sigma}{2} \nabla \log \rho_t -\proj_{\rho_t}(b)|^2(x) \rho_t(x) \De x \De t, \\&&\partial_t \rho_t+\nabla\cdot(v_t \rho_t)=0, \quad v_t \in \tsp_{\rho_t}\label{FDproblem2gioaux3}\\&& \rho_0=\mu_0, \quad \rho_1=\mu_1,\label{FDproblem3gioaux3}
\end{eqnarray}
\end{subequations}
The conclusion now follows from the definitons  \eqref{eq:calB},\eqref{eq:shannongrad} 
\end{proof}

\begin{lemma}\label{lem:2ndFDformulation}
	The formulation \eqref{eq:FDproblemgio} is equivalent to 
\begin{subequations}\label{eq:FDproblemgio2}
\begin{eqnarray}\label{FDproblem1gio2}&&\inf_{(\rho,v)} \frac{1}{2}\int_{0}^1 |v_t|^2_{\tsp_{\rho_t}}+  \frac{\sigma^2}{4}\mathcal{I}(\rho_t)+|\mathscr{B}(\rho_t)|^2_{\tsp_{\rho_t}} +\mathcal{E}_{\sigma \nabla \cdot b} (\rho_t)- 2\langle \mathscr{B}(\rho_t),v_t \rangle_{\tsp_{\rho_t}}   \De t	 \\&&\partial_t \rho_t+\nabla\cdot(v_t\rho_t)=0, \quad v_t \in \tsp_{\rho_t} \label{FDproblem2gio2}\\&& \rho_0=\mu_0, \quad \rho_1=\mu_1,\label{FDproblem3gio2}
\end{eqnarray}
\end{subequations}

\end{lemma}

\begin{proof}
It suffices to expand the squares, use the defintion \eqref{eq:fish}, observe that the product $-\sigma\int_{0}^1 \langle \otgrad \mathcal{S}(\rho_t) ,v_t\rangle_{\tsp_{\rho_t}}\De t$ is constantly equal to $\mathcal{S}(\mu_0)-\mathcal{S}(\mu_1)$ over the admissible region and can thus be neglected, and  to rewrite the cross product  $\sigma\int_0^1 \langle \otgrad \mathcal{S}(\rho_t),\mathscr{B}(\rho_t) \rangle_{\tsp_{\rho_t}} \De t$ in the following way:
\begin{eqnarray*}
&{}& \sigma\int_0^1 \langle \otgrad \mathcal{S}(\rho_t),\mathscr{B}(\rho_t) \rangle_{\tsp_{\rho_t}} \De t\\
	&=&-\sigma\int_0^1 \int_{\R^d} \langle \nabla \log \rho_t(x), b(x) \rangle	\rho_t(x) \De x \De t\\
&=& - \sigma \int_{0}^{1} \int_{\R^d} \langle \nabla \rho_t, b \rangle (x) \De x \De t \\
&=& \sigma \int_{0}^{1}\int_{\R^d} (\nabla \cdot b)\rho_t(x) \De x \De t \\
&\stackrel{\eqref{eq:energy}}{=}& \int_{0}^{1}\mathcal{E}_{\sigma  \nabla \cdot b}(\rho_t) \De t.
\end{eqnarray*}

\end{proof}

\section{Variational analysis}\label{VA}
The goal of this section is to prove the following result.
\begin{theorem}\label{thm:2ndordereqSB}
	Let $(\mu_t, 0\le t\le 1)$ be the entropic interpolation between $\mu$ and $\nu$ and $(v_t)$ its velocity field. Then $(\mu_t)$ satisfies the equation
	\begin{equation}\label{eq:2ndordereqSB}
		\covdev_{v_t} v_t =  \otgrad \left[\frac{\sigma}{8}\mathcal{I}+ \frac{1}{2}|\mathscr{B}|^2_{\tsp_{\cdot}}  + \frac{1}{2}\mathcal{E}_{\sigma \nabla \cdot b}\right] (\mu_t) +  2\,\visc_{v_t} \mathscr{B}
	\end{equation}
	
\end{theorem}
In the reversible case, i.e. when $b = -\nabla U$, we obtain the following

\begin{cor}\label{cor:2ndordereqSBrev}
Let $b= -\nabla U$, $(\mu_t)$ be the entropic interpolation between $\mu$ and $\nu$ and $(v_t)$ its velocity field. Then $(\mu_t)$ satisfies the equation
\begin{equation}\label{eq:reversibleSB}
\covdev_{v_t} v_t =  \otgrad \left[\frac{\sigma}{8}\mathcal{I}  + \frac{1}{2}\mathcal{E}_{|\nabla U|^2-\sigma \Delta U }\right] (\mu_t) 
\end{equation}
\end{cor}
Note that, in the above corollary, the quantity $|\nabla U|^2-\sigma \Delta U $ is the \emph{reciprocal characteristic} associated to the potential $U$. Among others Levy, Krener and Thieullen \cite{LeKr,Kr,Th} observed that this quantity should express a kind of mean acceleration for the Schr\"odinger bridge, or, more generally, for reciprocal processes.  Equation \eqref{eq:reversibleSB} clearly shows that this is the case.
Notice that the covariant derivative has also been computed in \cite[p.6]{GT} without relating it to the gradient of the Fisher information functional. Prior to the proof, let us introduce some useful terminology. In analogy to the classical calculus of variations, we call a \emph{variation} of $(\mu_t)$ a family of curves $(\rho^s_t)$ such that 
\begin{itemize}
\item $(\rho^0_{t})=(\mu_t)$
\item $ \forall s\in (-\varepsilon,\varepsilon), \quad (\rho^s_0) = \mu_0$ and $\rho^s_1= \mu_1$.
\end{itemize}
For any fixed $s$, $v^s_t$ stands for the velocity field of the curve $t \mapsto \rho^s_{t}$. That is
\begin{equation}\label{eq:vst}
\partial_t \rho^s_t + \nabla \cdot( v^s_t \rho^s_t) =0,
\end{equation}
and $v^s_t$ is a gradient vector field. In particular, we have that $(v^0_t)=(v_t)$, where $v_t$ is the velocity field of $(\mu_t)$.
On the other hand, for fixed $t$,  $w^s_t$ is the velocity field of the curve $s \mapsto \rho^s_t$, i.e.
\begin{equation}\label{eq:wst}
\partial_s \rho^s_t + \nabla \cdot(w^s_t \rho^s_t) =0,
\end{equation}
and $w^s_t$ is a gradient vector field. 

\begin{proof}
	
Consider a variation $(\rho^s_t)$ of $(\mu_t)$. Imposing the optimality condition of $(\mu_t)$ in \eqref{eq:FDproblemgio2} yields
\begin{equation}\label{eq:criticality1}
\int_0^1 \partial_s \left( \frac{1}{2}|v^s_t|^2_{\tsp_{\rho^s_t}} + \frac{\sigma^2 }{8}\mathcal{I}(\rho^s_t)+  \frac{1}{2}|\mathscr{B}|^2_{\tsp_{\rho^s_t}} + \frac{1}{2} \mathcal{E}_{\sigma \nabla \cdot b}(\rho^s_t) - \frac{1}{2}\langle v^s_t, \mathscr{B}(\rho^s_t) \rangle_{\tsp_{\rho^s_t}} \right)\Big|_{s=0} \De t \leq 0
\end{equation}
For convenience in the analysis, set
\begin{eqnarray*}
A &:=&	\int_0^1 \partial_s \left( \frac{1}{2}|v^s_t|^2_{\tsp_{\rho^s_t}} \right)\Big|_{s=0} \De t \\
B &:=& \int_0^1\partial_s   \left( \frac{\sigma^2 }{8} \mathcal{I}(\rho^s_t) + |\mathscr{B}|^2_{\tsp_{\rho^s_t}}+\frac{1}{2} \mathcal{E}_{\sigma \nabla \cdot b}(\rho^s_t) \right) \Big|_{s=0} \De t\\
C &:=& -\int_{0}^1 \partial_s \left( \langle v^s_t, \mathscr{B}(\rho^s_t) \rangle_{\tsp_{\rho^s_t}} \right)\Big|_{s=0} \De t
\end{eqnarray*}

Let us work separately on each term.
\paragraph*{Term A} Using the product rule we obtain

\begin{equation}\label{eq:firstterm1}
 \frac{1}{2}\int_0^1 \partial_s |v^s_t|^2_{\tsp_{\rho^s_t}}\De t= \underbrace{\int \langle v^s_t,\partial_s v^s_t \rangle \rho^s_t   \De x \De t}_{:=A.1}+ \underbrace{\frac{1}{2}\int |v^s_t|^2  \partial_s \rho^s_t  \De x  \De t}_{:=A.2},
\end{equation}
where use the shorthand notation  $\int$ for $\int_{0}^1 \int_{\R^d}$. We keep doing this throughout the proof. Since $v^s_t$ is a gradient vector field, we have that $v^s_t = \nabla \varphi^s_t$ for some potential $\varphi^s_t$. With a first integration by parts in space we get
\begin{equation}\label{eq:criticality2} A.1=-\int \varphi^s_t \nabla \cdot\left( \partial_s v^s_t  \rho^s_t\right)  \De t \De x \end{equation}Next, we observe that the following identity can be obtained differentiating \eqref{eq:vst} with respect to $s$ and using \eqref{eq:wst}
\begin{equation}\label{eq:dsdt}
\partial_s \partial_t \rho^s_t + \nabla \cdot ( v^s_t\,  \partial_s \rho^s_t+\ \partial_s v^s_t \,\rho^s_t) =0.  
\end{equation}
The identity \eqref{eq:dsdt} allows us to rewrite  \eqref{eq:criticality2} as
\begin{equation}\label{eq:criticality7} A.1= \underbrace{ \int \varphi^s_t \, \partial_t \partial_s \rho^s_t \De x \De t}_{A.1.1} +\underbrace{ \int \varphi^s_t\,  \nabla \cdot\left(  v^s_t  \partial_s\rho^s_t\right)  \De x \De t}_{A.1.2}  
\end{equation} 
Using subsequently the continuity equations \eqref{eq:vst},\eqref{eq:wst} and integration by parts we have
\begin{eqnarray*}
A.1.1 &=& \int \partial_t \varphi^s_t \,  \nabla \cdot \left( w^s_t \rho^s_t \right) \De x \De t  \\
 &=& -\int \langle \partial_t v^s_t ,   w^s_t  \rangle \rho^s_t \De x \De t,
\end{eqnarray*}
where we used the fact that $\partial_s \rho_0^s = \partial_s \rho^s_1 =0$ 
to get rid of the boundary terms when integrating by parts in time. Moreover, integrating by parts in space and using \eqref{eq:wst}
\begin{eqnarray*}
A.1.2 &=& \int | v^s_t|^2  \nabla \cdot \left(w^s_t \rho^s_t \right)  \De x \De t \\
&=& -\int \langle \nabla \left( | v^s_t|^2\right) , w^s_t \rangle \rho^s_t  \De x \De t
\end{eqnarray*}
Plugging the newly obtained expressions for  $A.1.1$ and $A.1.2$ back into \eqref{eq:criticality7} we obtain
\[ A.1 =-\int \langle   \partial_t v^s_t+ \nabla \left( | v^s_t|^2\right) , w^s_t \rangle \rho^s_t \De x \De t \]
Concerning the term $A.2$ in  \eqref{eq:firstterm1}, using  \eqref{eq:wst} and an integration by parts we get
\[A.2 =\frac{1}{2}\int \langle \nabla \left(|v^s_t|^2\right) , w^s_t \rangle \rho^s_t  \De x \De t  \]
so that the whole expression in \eqref{eq:firstterm1} equals
\begin{equation}\label{eq:criticality6}
-\int \langle   \partial_t v^s_t+ \frac{1}{2}\nabla\left( | v^s_t|^2\right) , w^s_t \rangle \rho^s_t \De x  t \end{equation}
Evaluating at $s=0$, and using \eqref{eq:acc}, we obtain
\begin{equation}\label{eq:firstterm2} 
A=-\int \langle\covdev_{v_t} v_t, w^0_t \rangle_{\tsp_{\mu_t}} \De t
\end{equation}
\paragraph*{Term B}
For this term, we have, using \eqref{eq:otgrad}:

\begin{equation}\label{eq:crititcality3}
B= \int_0^1 \Big\langle\otgrad  \left[ \frac{\sigma^2 }{8} \mathcal{I} +  \frac{1}{2}|\mathscr{B}|^2_{\tsp_{\cdot}}+\mathcal{E}_{\sigma \nabla \cdot b} \right](\mu_t) , w^0_t \Big\rangle_{\tsp_{\mu_t}} \De t.
\end{equation}

\paragraph*{Term C}
First observe that we can rewrite it as
\begin{equation}\label{eq:fourthterm1}   -\int \langle \partial_s v^s_t|_{s=0}, \mathscr{B}(\mu_t)\rangle \mu_t + \langle  v_t,  \partial_s \mathscr{B}(\rho^s_t)|_{s=0}\rangle \mu_t + \langle  v_t,   \mathscr{B}(\mu_t)\rangle \partial_s\rho^s_t|_{s=0} \De x \De t  
\end{equation}

Let us focus on the first summand. Since $\mathscr{B}(\mu_t)$ is a gradient vector field, we may write it in the form $\nabla U_t$ for some potential $U$. Thus we have, using integration by parts, 
\begin{eqnarray*}
&{}&-\int \langle \partial_s v^s_t|_{s=0}, \mathscr{B}(\mu_t)\rangle \mu_t  \De x \De t\\
&=&-\int \langle \partial_s v^s_t|_{s=0}, \nabla U_t \rangle \mu_t \De x  \De t\\
&=&\int  \nabla \cdot(\partial_s v^s_t|_{s=0}\mu_t ) U_t \De x \De t . \\
\end{eqnarray*}
Taking advantage of \eqref{eq:dsdt}, we can write the last integral as
\begin{eqnarray*}
&{}&-\int  \partial_s \partial_t \rho^s_t|_{s=0} \, U_t \, \De x \De t  -\int  \nabla \cdot(\partial_s \rho^s_t|_{s=0} \, v_t) U_t \,  \De x  \De t\\
&\stackrel{\eqref{eq:wst}}{=}& \int   \partial_t (\nabla \cdot ( \mu_t w^0_t) ) U_t  \, \De x \De t + \int  \langle v_t, \mathscr{B}(\mu_t)  \rangle\,   \partial_s \rho^s_t|_{s=0}\,  \De x \De t\\
&=&\int   \langle w^0_t, \partial_t \mathscr{B}(\mu_t)\rangle \mu_t \De x \De t+\int  \langle v_t, \mathscr{B}(\mu_t)  \rangle   \partial_s \rho^s_t|_{s=0} \De x\De t  
\end{eqnarray*}
Note now that the second term of this last expression cancels with the third term  in \eqref{eq:fourthterm1}. Therefore the whole expression in \eqref{eq:fourthterm1} equals
\begin{equation}\label{eq:criticality4}\int \left(  \langle w^0_t, \partial_t \mathscr{B}(\mu_t)\rangle - \langle  v_t,  \partial_s \mathscr{B}(\rho^s_t)|_{s=0}\rangle   \right) \mu_t  \De x \De t. 
 \end{equation}

Next, since $\mathscr{B}(\rho^s_t)$ is a gradient vector field we have
\[\int \langle w^0_t, \mathrm{D}_{v_t}\mathscr{B}(\mu_t)  \rangle \mu_t  \De x \De t = \int \langle  v_t , \mathrm{D}_{w^0_t}\mathscr{B}(\mu_t) \rangle \mu_t  \De x \De t \] 
and therefore we can rewrite \eqref{eq:criticality4} as
\begin{eqnarray*} 
&{}&\int   \left( \Big\langle w^0_t, \partial_t \mathscr{B}(\mu_t)+\mathrm{D}_{v_t} \mathscr{B}(\mu_t)\Big\rangle - \Big\langle  v_t,  \left(\partial_s \mathscr{B}(\rho^s_t)+\mathrm{D}_{w^s_t}\mathscr{B}(\rho^s_t) \right)\big|_{s=0}\Big\rangle  \right) \mu_t  \De x \De t\\
&=&\int   \left( \Big\langle w^0_t, \proj_{\mu_t} \left(\partial_t \mathscr{B}(\mu_t)+\mathrm{D}_{v_t} \mathscr{B}(\mu_t) \right) \Big\rangle - \Big\langle  v_t,  \proj_{\mu_t}\left(\partial_s \mathscr{B}(\rho^s_t)+\mathrm{D}_{w^0_t}\mathscr{B}(\mu_t) \right)\Big\rangle  \right) \mu_t  \De x \De t\\
&\stackrel{\eqref{eq:covdedv}}{=}& \int \langle w^0_t ,\otJac_{v_t} \mathscr{B} \rangle_{\tsp_{\mu_t}} - \langle v_t ,\otJac_{w^0_t} \mathscr{B} \rangle_{\tsp_{\mu_t}} \De t \\
&\stackrel{\eqref{eq:visc}}{=}& 2\int \langle w^0_t ,  \, \visc_{v_t} \mathscr{B} \rangle_{\tsp_{\mu_t}  }\De t.
\end{eqnarray*}
Thus 
\begin{equation}\label{eq:criticality5}
C=2 \int \langle w^0_t ,  \, \visc_{v_t} \mathscr{B} \rangle_{\tsp_{\mu_t}  }\De t
\end{equation}
\newline
Putting togehter \eqref{eq:criticality6},\eqref{eq:crititcality3},\eqref{eq:criticality5}
we can rewrite the criticality condition \eqref{eq:criticality1} as
\begin{equation*}
\int_0^1 \Big\langle -\covdev_{v_t} v_t+ \otgrad \left[ \frac{\sigma^2 }{8} \mathcal{I} +  \frac{1}{2}|\mathscr{B}|^2_{\tsp_{\cdot}} +\frac{1}{2} \mathcal{E}_{\sigma \nabla \cdot b} \right](\mu_t)+ 2 \,\visc_{v_t} \mathscr{B} , w^0_t \Big\rangle_{\tsp_{\mu_t}} \De t \leq 0
\end{equation*}
The desired conclusion now follows from the fact that $w^0_t$ can be chosen arbitrarily among the gradient vector fields.
\end{proof}
}
\section{Madelung fluid and stochastic mechanics}\label{MFSM}

\subsection{The Madelung Fluid}
At the very dawn of quantum mechanics, in 1926, Erwin Madelung proposed in \cite{M} a fluid dynamic formulation of quantum mechanics, see \cite{G} for a beautiful account. He noticed that if $\{\psi(x,t), t_0\le t\le t_1\}$ satisfies the Schr\"{o}dinger equation 
\begin{equation}\label{SE}
\frac{\partial{\psi}}{\partial{t}} =
\frac{i\hbar}{2m}\Delta\psi -
\frac{i}{\hbar}V(x)\psi,
\end{equation}
then, writing $\psi(x,t)=\exp [R(x,t)+\frac{i}{\hbar}S(x,t)]$, $\rho_t(x)=|\psi(x,t)|^2=\exp[2R(x,t)]$, we get for $R$ and $S$  the system of partial differential equations
\begin{subequations}\label{eq:sa}
\begin{eqnarray}\label{sa1}
&&\frac{\partial R}{\partial{t}}+\frac{1}{m}\nabla R\cdot\nabla S+
\frac{1}{2m}\Delta S=0,\\\label{sa2}
&&\frac{\partial S}{\partial t}+\frac{1}{2m}\nabla S\cdot\nabla S+ V
-\frac{\hbar^2}{2m}\left[\nabla R\cdot\nabla R+\Delta R\right]=0.
\end{eqnarray}
\end{subequations}
This system represents the {\em Madelung fluid}. 
{If we set $ \mu_t(x) = \exp(2 R(x,t))$ and $v_t = \frac{1}{m}\nabla S_t$, we can rewrite the system above, after taking the spatial gradient in the second equation and performing some standard calculations, in the following way
\begin{subequations}\label{eq:madelungfluid2nd}
\begin{eqnarray}\label{madelungfluid2nd1}
&&\partial_t \mu_t +\nabla \cdot( v_t \rho_t )=0,\\\label{madelungfluidsecond2}
&&\partial_t v_t+\frac{1}{2}\nabla|v_t|^2 +\frac{1}{m}\nabla V
-\frac{\hbar^2}{8m^2}\nabla \left[|\nabla \log \mu_t|^2 +2 \Delta \log \mu_t \right]=0.
\end{eqnarray}
\end{subequations}
Using \eqref{eq:acc},\eqref{eq:fishgrad},\eqref{eq:energygrad} we obtain that the PDE system above has the clear interpretation of a Newton's law, when $(\mu_t)$ is seen as a curve in the Riemannian manifold of optimal transport:
\[     \covdev_{v_t} v_t = -\frac{\hbar^2}{8m^2} \otgrad \mathcal{I}(\mu_t)-\frac{1}{m} \otgrad \mathcal{E}_V(\mu_t)     \]
This fact was first understood by Von Renesse \cite{vonR} after some early efforts to connect OMT to stochastic mechanics \cite{Carlen}, \cite[p.707]{Vil2}. Note that this equation has the same form of \eqref{eq:reversibleSB}, with the only (remarkable) difference of the minus sign in front of the gradient of the Fisher information.
}
This variational fluid-dynamic formulation fully clarifies the relation between Shr\"{o}dinger Bridges and Nelson's stochastic mechanics (see next subsection) which has  puzzled mathematical physicists since Schr\"{o}dinger in \cite{S1}\footnote{ Schr\"{o}dinger writes: ``{\em  Merkw\"{u}rdige Analogien zur Quantenmechanik, die mir sehr des Hindenkens wert erscheinen}" (remarkable analogies to quantum mechanics which appear to me very worth of reflection).}, see also \cite{Z}. Notice that the Lagrangian with the plus sign also plays a role in Nelson's stochastic mechanics: It appears in the particle (first form) of Hamilton's-like principle due to Yasue \cite{Y,N2}) which relies on stochastic calculus. The Lagrangian with the minus sign in front of the Fisher information (Guerra-Morato Lagrangian) appears instead in the second, fluid dynamic formulation \cite{GM,N2}. This has contributed in the past  to the confusion on the relation between the two theories which we hope this paper has definitely resolved. Historically, equation  (\ref{sa1})   was recognized already by Madelung in 1926 to be just a continuity equation. The second was thought to be a Hamilton-Jacobi equation featuring, besides the classical potential $V$, the ``quantum potential" \cite{Bo}\footnote{It is truly unfortunate that Bohm, who had related such quantum potential to the gradient of the Fisher information functional \cite[p.172]{Bo}, did not realized that this made postulating such a non physical potential unnecessary. This has, among other things, led to an endless controversy between Bohmian and Nelsonian  followers.}
\[-\frac{\hbar^2}{2m}\left[\nabla R\cdot\nabla R+\Delta R\right]=-\frac{\hbar^2}{2m}\frac{\Delta
\sqrt{\rho}(x,t)}{\sqrt{\rho_t(x)}}.
\]
The presence of such a ``nonlocal" potential, however, mystified researchers for a long time. \subsection{Remarks on stochastic mechanics}
The first formulation of quantum mechanics through stochastic processes is due to the Hungarian theoretical physicist Imre Fenyes in a series of papers between 1946 and 1952. In the abstract of \cite{F},  we find the following impressive statement:``Die wellenmechanischen
Prozesse sind spezielle MARKOFFsche Prozesse. Die Relation von Heisenberg  ist (im Gegensatz zur bisherigen Interpretation) ausschliesslich eine Folge der statistischen Betrachtungsweise, und ist von den bei den Messungen auftretenden St\"{o}rungen unabh\"{a}ngig.\footnote{``The wave mechanical
Processes are special MARKOFF processes. The relation of Heisenberg (in contrast to the previous interpretation) is exclusively a consequence of the statistical approach, and is independent of the disturbances occurring during the measurements."}"  After some further contributions by Weizel in the fifties \cite[p.135]{N1}, Edward Nelson provided an elegant formulation of stochastic mechanics in 1966-67 \cite{N0,N1} by developing a new kinematics for continuous stochastic evolutions.  The Schr\"{o}dinger
equation was originally derived from the continuity equation (\ref{sa2}) plus a Newton type
law. The Newton-Nelson law was later shown to follow, in analogy
to classical mechanics, from a Hamilton-like stochastic variational 
principle \cite{Y,GM}.
Other versions of the variational principle have been proposed in
\cite{N2,BCZ,M2,P0,ROS}. One of the most striking differences between 
Nelson's and other
versions of mechanics such as  Bohmian mechanics
\cite{Bo,BV,BH} or the Levy-Krener mechanics \cite{LK,LK2} is that 
the former features a
kinematics for finite-energy diffusions with {\em two} velocities. Indeed, let
$\{\psi(x,t); t_0\le t\le t_1\}$, satisfy (\ref{SE})
be such that
\begin{equation}\label{FA}||\nabla\psi||^2_2\in
L^1_{{\rm loc}}[t_0,+\infty).
\end{equation}
This is Carlen's
{\em finite action condition}. Under these hypotheses, the Nelson measure
$P$
may be constructed on path space, \cite{C},\cite{Car}, \cite [Chapter
IV]{BCZ}, and references therein.
Namely, letting  $\Omega:={\cal C}([t_0,t_1],\R^d)$ the $n$-dimensional
continuous functions on $[t_0,t_1]$, under the  probability measure $P$,
the canonical coordinate process $x(t,\omega)=\omega(t)$ is an
$n$-dimensional, Markov,
finite-energy diffusion process $\{x(t);t_0\le t\le t_1\}$,
called {\em Nelson's process}, having
forward and backward Ito differentials
\begin{subequations}\label{eq:diff}
\begin{eqnarray}\label{diff1}
dx=\left[\frac{\hbar}{m}\nabla\left(
R+ \frac{1}{\hbar}S\right)\right](x(t),t)dt
+\sqrt{\frac{\hbar}{m}}dw_+(t),\\ \label{diff1}
dx=\left[\frac{\hbar}{m}\nabla\left(
-R+ \frac{1}{\hbar}S\right)\right](x(t),t)dt
+\sqrt{\frac{\hbar}{m}}dw_-(t),
\end{eqnarray}
\end{subequations}
where $w_+,w_-$ are standard, $n$-dimensional Wiener process. Moreover, the
probability density $\rho(\cdot,t)$ of
$x(t)$ satisfies Born's relation
\begin{equation}\label{D}\rho_t(x)=|\psi(x,t)|^2,\quad \forall t \in [t_0,t_1],
\end{equation}
and 
\[v_t(x)=\frac{1}{m}\nabla S(x,t), \quad u(x,t)=\frac{\hbar}{m}\nabla R(x,t)\]
are the current and osmotic velocity fields. It should be emphasized, that Nelson's stochastic mechanics, when compared to other alternative formulations, features the following advantages: There is no need to postulate any potential besides the physical potential $V$, but only to adopt what is the natural kinematics for finite energy diffusions \cite{Foe, F2}. Notice that  such a kinematics leads naturally to the complexification of the velocity and of the momentum \cite{P0,P01, KGP}. The latter fact permits to develop Hamilton-like equations and provides a {\em kinematical} justification (in constrast to the countless  {\em mathematical} justifications)  for the need of complex quantities in quantum mechanics. For a list of further ``successes", as well as failures, of Nelson's stochastic mechanics see \cite{N4}. Among the successes, it might be worthwhile quoting from \cite{N4} the little known: ``A stochastic picture of the two-slit experiment, explaining how particles have trajectories
going through just one slit or the other, but nevertheless produce a probability density as for
interfering waves; see \cite[Section 17]{N2}", see also \cite{P3}.

Although both the hydrodynamical formulation and the stochastic formulation are alternative to usual quantum mechanics, they are not equivalent. Indeed, let the pair $(v^*,\rho^*)$ be optimal for Problem \ref{FDproblem} and consider the following random evolution
\[\dot{X}(t)=v^*(X(t),t), \quad X(0)\sim
 \rho_0dx,
\]
where we have assumed that $v^*$ guarantees existence and uniqueness of the initial value problem for $t\in [0,1]$ and any deterministic initial condition. Then the probability density $\rho_t(x)$ of $X(t)$ does satisfies (weakly) (\ref{SBB2}) with the same initial condition and therefore coincides with $\rho^*(x,t)$. This shows that variational principles on Wasserstein space such as those of this paper do not encode the same information as those formulated in terms of stochastic processes. In particular, the information on the diffusion coefficient - smoothness of trajectories is lacking. While this might not be so relevant in some applications such as image interpolation (morphing) where the flow of one time marginals is the object of interest, it is definitely crucial in physics. In stochastic mechanics, for instance, the assumption that the evolution be described by a diffusion process with diffusion coefficient $\frac{\hbar}{m}$ is made at the onset leading to stochastic control problems and corresponding Hamilton-Jacobi-Bellman equations. 

Similar considerations can be made in the case of a general Markovian prior treated in Section \ref{SOC}. Within Nelson's stochastic mechanics, this would correspond to the variational mechanism leading to the so-called collapse of the wavefunction after a position measurement \cite{P1,P2}. The zero-noise limit of a general Schr\"{o}dinger bridge, instead, is given by  an OMT problem with prior formulated and studied in \cite{CGP3}, see also \cite{CGP4}.

\end{document}